%% file: 000-paper.tex
\PassOptionsToPackage{table}{xcolor}
\documentclass{article}

\usepackage{microtype}
\usepackage{graphicx}
\usepackage{subfigure}
\usepackage{booktabs} 

\usepackage{pifont}
\usepackage{multirow}
\usepackage[most]{tcolorbox}

\definecolor{light-gray}{gray}{0.8}
\graphicspath{{./figs/}}

\usepackage{hyperref}

\usepackage{icml2024/algorithm}
\usepackage{icml2024/algorithmic}
\usepackage{icml2024/fancyhdr}

\usepackage[accepted]{icml2024/icml2024}

\usepackage{amsmath}
\usepackage{amssymb}
\usepackage{mathtools}
\usepackage{amsthm}

\usepackage[capitalize,noabbrev]{cleveref}

\theoremstyle{plain}

\newtheorem{MyTheo}{Theorem}
\newtheorem{MyAppLemma}{Lemma}

\theoremstyle{definition}

\newtheorem{MyDef}{Definition}
\newtheorem{MyAssump}{Assumption}

\theoremstyle{remark}

\newcommand{\review}[1]{{\color{blue}{#1}}}

\usepackage[textsize=tiny]{todonotes}

\usepackage{xspace}

\input{my_macros}
\input{math_commands}

\icmltitlerunning{PuriDefense: Randomized Local Implicit Adversarial Purification}

\begin{document}

\twocolumn[
    \icmltitle{PuriDefense: Randomized Local Implicit Adversarial Purification \\
        for Defending Black-box Query-based Attacks}



    \icmlsetsymbol{equal}{*}

    \begin{icmlauthorlist}
        \icmlauthor{Ping Guo}{yyy}
        \icmlauthor{Xiang Li}{xxx}
        \icmlauthor{Zhiyuan Yang}{yyy}
        \icmlauthor{Xi Lin}{yyy}
        \icmlauthor{Qingchuan Zhao}{yyy}
        \icmlauthor{Qingfu Zhang}{yyy}
    \end{icmlauthorlist}

    \icmlaffiliation{yyy}{City University of Hong Kong, Hong Kong}
    \icmlaffiliation{xxx}{Southeast University, China}

    \icmlcorrespondingauthor{Ping Guo}{pingguo5-c@my.cityu.edu.hk}

    \icmlkeywords{Adversarial Attacks}

    \vskip 0.3in
]



\printAffiliationsAndNotice{}  

\begin{abstract}
    Black-box query-based attacks constitute significant threats to Machine Learning as a Service (MLaaS) systems since they can generate adversarial examples without accessing the target model's architecture and parameters. Traditional defense mechanisms, such as adversarial training, gradient masking, and input transformations, either impose substantial computational costs or compromise the test accuracy of non-adversarial inputs. To address these challenges, we propose an efficient defense mechanism, PuriDefense, that employs random patch-wise purifications with an ensemble of lightweight purification models at a low level of inference cost. These models leverage the local implicit function and rebuild the natural image manifold. Our theoretical analysis suggests that this approach slows down the convergence of query-based attacks by incorporating randomness into purifications. Extensive experiments on CIFAR-10 and ImageNet validate the effectiveness of our proposed purifier-based defense mechanism, demonstrating significant improvements in robustness against query-based attacks.
\end{abstract}

\input{001-intro}
\input{002-back}
\input{003-method}
\input{004-evaluation}

\section{Conclusion}
This paper introduces a novel theory-backed image purification mechanism utilizing local implicit function to defend deep neural networks against query-based adversarial attacks. The mechanism enhances classifier robustness and reduces successful attacks whilst also addressing vulnerabilities of deterministic transformations. Its effectiveness and robustness, which increase with the addition of purifiers, have been authenticated via extensive tests on CIFAR-10 and ImageNet. Our work highlights the need for dynamic and efficient defense mechanisms in MLaaS systems.



\newpage
\section*{Impact Statement}

This paper presents work whose goal is to advance the field of Machine Learning. There are many potential societal consequences of our work, none which we feel must be specifically highlighted here.

\bibliography{000-paper}
\bibliographystyle{icml2024/icml2024}

\newpage
\appendix
\onecolumn
\input{006-app}



\end{document}

%% file: my_macros.tex
\newcommand{\figref}[1]{\figurename~\ref{#1}}
\newcommand{\tabref}[1]{\tablename~\ref{#1}}

\newcommand{\secref}[1]{Section~\ref{#1}}

\newcommand{\bheading}[1]{{\vspace{0.3\baselineskip}\noindent{\textbf{#1}}}}
\newcommand{\lheading}[1]{{\vspace{0.3\baselineskip}\noindent{\textit{#1}}}}

\newcommand{\eg}{\textit{e.g.}\xspace}
\newcommand{\ie}{\textit{i.e.}\xspace}

%% file: math_commands.tex
\usepackage{amsmath,amsfonts,bm}

\def\figref#1{figure~\ref{#1}}
\def\Figref#1{Figure~\ref{#1}}

\def\secref#1{section~\ref{#1}}

\def\eqref#1{equation~\ref{#1}}

\def\1{\bm{1}}

\def\rk{{\textnormal{k}}}

\def\ru{{\textnormal{u}}}

\def\rvu{{\mathbf{i}}}

\def\rvk{{\mathbf{k}}}

\def\rvu{{\mathbf{u}}}

\def\rmU{{\mathbf{U}}}

\def\vm{{\bm{m}}}

\def\vx{{\bm{x}}}
\def\vy{{\bm{y}}}

\DeclareMathAlphabet{\mathsfit}{\encodingdefault}{\sfdefault}{m}{sl}
\SetMathAlphabet{\mathsfit}{bold}{\encodingdefault}{\sfdefault}{bx}{n}



%% file: 001-intro.tex
\section{Introduction\label{sec:intro}}
Deep neural networks (\textit{DNNs}), while presenting remarkable performance across various applications, are mostly leaning to become subject to \emph{adversarial attacks}, where even slight perturbations to the inputs can severely compromise their predictions~\citep{szegedy:2014:intriguing}.
This notorious vulnerability significantly challenges the inherent robustness of DNNs and could even make the situation much worse when it comes to security-critical scenarios, such as facial recognition~\citep{deng:2019:efficient} and autonomous driving~\citep{cao:2019:adversarial}.
Accordingly, attackers have devised both \emph{white-box attacks} if having full access to the DNN model and \emph{black-box attacks} in case the model is inaccessible.
While black-box attacks appear to be more challenging, it is often considered a more realistic threat model, and its state-of-the-art (SOTA) could leverage a limited number of queries to achieve high successful rates against closed-source commercial platforms, \ie Clarifai~\citep{clarifai} and Google Cloud Vision API~\citep{google_vision}, presenting a disconcerting situation~\cite{liu:2017:delving}.

Defending black-box query-based attacks in real-world large-scale Machine-Learning-as-a-Service (\textit{MLaaS}) systems calls for an extremely low extra inference cost. This is because business companies, such as Facebook, handle millions of image queries daily and thereby increase the extra cost for defense a million-fold~\citep{facebook}. This issue prohibits testing time defenses to run multiple inferences to achieve \emph{certified robustness}~\citep{cohen:2019:certified,salman:2020:denoised}. Moreover, training time defenses, \ie retraining the DNNs with large datasets to enhance their robustness against adversarial examples (\eg \emph{adversarial training}~\citep{madr:2018:towards} and \emph{gradient masking}~\citep{florian:2018:ensemble}), impose substantial economic and computational costs attributed to the heavy training expense. Therefore, there is a critical need for a lightweight yet effective strategy to perform adversarial purifications to enable low inference cost for achieving robustness.

Given the aforementioned challenges, recent research efforts have been devoted to either eliminating or disturbing adversarial perturbations prior to the forwarding of the query image to the classifier. Nevertheless, the existing methods that include both heuristic transformations and neural network-based adversarial purification models have certain limitations in removing adversarial perturbations. While heuristic transformation methods cause minimal impact on cost, they merely disrupt adversarial perturbations and often negatively impact the testing accuracy of non-adversarial inputs~\cite{xu:2018:feature, qin:2021:random}. Moreover, neural network-based purifications aiming to completely eradicate adversarial perturbations can even exceed the computational burden of the classifier itself~\cite{carlini:2023:certified}. Consequently, there have been no effective defense mechanisms that can achieve both high robustness and low computational cost against black-box query-based attacks.

In this study, we present PuriDefense, a novel random patch-wise image purification mechanism that leverages local implicit functions to enhance the robustness of classifiers against query-based attacks. Local implicit functions, initially developed for super-resolution applications~\citep{lim:2017:edsr,zhang:2018:image}, have shown potential in efficiently mitigating white-box attacks~\citep{ho:2022:disco}. Nonetheless, our empirical analysis indicates that a naively integrated local implicit function with a classifier yields a system vulnerable to query-based attacks, achieving only a \textbf{5.1\%} robust accuracy on the ImageNet dataset in potent attack scenarios. Our theoretical examination attributes this vulnerability to the absence of inherent randomness in the purification process.

While randomness can be incorporated through a collection of diverse purifiers~\citep{ho:2022:disco}, this approach linearly increases the inference cost. In contrast, PuriDefense enhances diversity without escalating inference costs, utilizing a pool of varied purifiers and assigning them to local image patches randomly. Our analysis confirms that with a broader spectrum of purifiers, the system exhibits higher resilience and effectively slows down the convergence of the query-based attacks.



Our contributions are summarized as follows:
\begin{itemize}
    \item We propose a novel defense mechanism using the local implicit function to randomly purify adversarial image patches using multiple purification models while maintaining the inference cost of a single model. Our work is the first to extend the local implicit function to defend against query-based attacks.
    \item We present a theoretical analysis illustrating the defense mechanism's effectiveness based on the convergence of black-box attacks. Our theoretical analysis suggests the robustness of our system escalates with the number of purifiers in the purifier pool.
    \item Our theoretical investigation further reveals the susceptibilities of \review{\textit{deterministic purifications}'} to query-based attacks and quantifies the enhanced robustness achieved by integrating randomness into any preprocessor-based defense strategy.
    \item We validate our method's defense capabilities through comprehensive experiments conducted on the CIFAR-10 and ImageNet datasets, specifically targeting SOTA query-based attacks.
\end{itemize}

\begin{figure*}[t]
    \centering
    \includegraphics[width=0.62\textwidth]{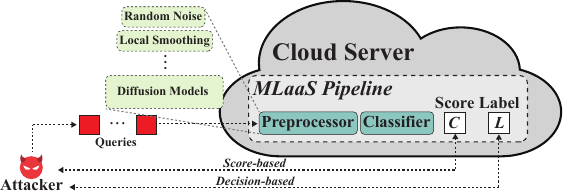}
    \vspace{-0.75\baselineskip}
    \caption{\small An illustration of the MLaaS system with preprocessor-based defense mechanisms under attack. The attackers can query the model with input $\vx$ and get the returned information $\mathcal{M}(\vx)$ which can be the confidence scores or the predicted label.\label{fig:ml_system}}
    \vspace{-1\baselineskip}
\end{figure*}

%% file: 002-back.tex
\section{Related Work\label{sec:back}}

\bheading{Query-based Attacks.}
Query-based attacks, continually querying the models to generate adversarial examples, are categorized as either \emph{score-based} or \emph{decision-based}, based on their access to confidence scores or labels, respectively.

Score-based attacks treat the MLaaS model, encompassing its pre-processors, the primary model, and post-processors, as an opaque entity. The typical objective function in this scenario is the marginal loss of confidence scores, illustrated in Equation~(\ref{eq:margin_loss}). To address the challenge of minimizing this loss, methods such as \emph{gradient estimation} and \emph{random search} are employed. \citet{ilyas:2018:blackbox} developed the pioneering limited-query score-based method, utilizing Natural Evolutionary Strategies (NES) for gradient estimation. This sparked a flurry of subsequent studies focusing on gradient estimation, methods like ZO-SGD~\citep{liu:2019:signsgd} and SignHunter~\citep{aldujaili:2020:signbits}. The forefront of score-based attack methods is represented by the Square attack~\citep{andriushchenko:2020:square}, which employs random search via localized square patch modifications and is often referenced as an important benchmark in model robustness assessments~\citep{croce:2021:robustbench}. Other algorithms that employ random search, such as SimBA~\citep{guo:2019:simba}, exist but do not achieve the effectiveness of the Square attack.

With respect to decision-based attacks, when confidence scores are not available, the available label information is used instead. \citet{ilyas:2018:blackbox} also pioneered work in this area, employing NES to optimize a heuristic proxy function with a limited number of queries. The gradient estimation method for decision-based attacks evolves to be more efficient by forming new optimization problems (e.g., OPT~\citep{cheng:2019:query}), and focusing on the gradient's sign rather than its magnitude (e.g., Sign-OPT~\citep{cheng:2020:signopt} and HopSkipJump~\citep{chen:2020:hop}). While direct search used in Boundary Attack~\citep{brendel:2018:decision} is the first decision-based attack, the HopSkipJump Attack is currently recognized as the most advanced attack.

\bheading{Adversarial Purification.}
The recent surge in the implementation of testing-time defenses, primarily for adversarial purification to enhance model robustness, is noteworthy. \citet{yoon:2021:adv} proposes the use of a score-based generative model to mitigate adversarial perturbations. Meanwhile, \citet{mao:2021:adv} employs self-supervised learning techniques, such as contrastive loss, to purify images. Following the success of diffusion models, purifications have been utilized to establish certified robustness for image classifiers~\citep{nie:2022:diffusion,carlini:2023:certified}. However, the complexity and vast number of parameters of diffusion models result in significantly reduced inference speeds of the classification systems.

The recent literature also notes the introduction of the local implicit function model as a defense against white-box attacks~\citep{ho:2022:disco}. Nonetheless, this approach has limitations, being trained only on white-box attacks and lacking theoretical guarantees for resisting black-box attacks. In contrast, our study restructures the network to omit multi-resolution support and accelerates the inference time by a factor of four in the implementation level. Furthermore, our defense mechanism design ensures that the inference speed remains constant as the number of purifier models increases, with a theoretical guarantee for resisting black-box attacks. Further information on general defense mechanisms is available for interested readers in Appendix~\ref{app:back}.

%% file: 003-method.tex
\section{Preliminaries\label{sec:preliminary}}
\subsection{Threat Model}
In the context of black-box query-based attacks, our threat model assumes that attackers possess have a limited comprehension of the target model. The interaction with the model, usually deployed on cloud servers, is confined to submitting queries and obtaining outputs as confidence scores or classifications. Attackers are deprived of deeper understanding of the model's internal mechanisms or the datasets employed. \Figref{fig:ml_system} depicts the MLaaS system under attack

\subsection{Query-based Attacks}
\subsubsection{Score-based Attacks}
Consider a classifier $\mathcal{M}:\mathcal{X} \to \mathcal{Y}$ deployed on a cloud server, where $\mathcal{X}$ denotes the input space and $\mathcal{Y}$ represents the output space. Attackers may query the model using an input $\vx \in \mathcal{X}$ and receive the corresponding output $\mathcal{M}(\vx) \in \mathcal{Y}$. When the model furnishes its output, typically as a confidence score, directly to the attackers, this situation typifies a score-based attack.

In this context, attackers craft an adversarial example $\vx_{adv}$ using the original example $\vx$ and its corresponding true label $y$. Their goal is to address the optimization challenge below, which is necessary to conduct an untargeted attack:
\begin{equation}\label{eq:margin_loss}
    \begin{aligned}
         & \min _{\vx_{adv} \in \mathcal{N}_R(\vx)} f(\vx_{adv}),                           \\
         & f(\vx_{adv})=\mathcal{M}_y(\vx_{adv})-\max _{j \neq y} \mathcal{M}_j(\vx_{adv}).
    \end{aligned}
\end{equation}
Here, $\mathcal{N}_R(\vx)=\left\{\vx' |\|\vx'-\vx\|_p\leq R\right\}$ represents the $\ell_p$-norm ball centered at $\vx$. For targeted attacks, the variable $j$ is specified as the designated target label, in contrast to being the index associated with the highest confidence score that is not the true label. An attack is considered successful if it results in the objective function falling below zero.

In white-box attacks, the projected gradient descent algorithm is employed; however, in a black-box setting, attackers lack access to gradient details. Consequently, black-box methods typically rely on two approaches to approximate the direction of function descent: \emph{gradient estimation} and \emph{heuristic search}. Additional information on these techniques can be found in Appendix~\ref{app:search}.

\subsubsection{Decision-based Attacks}
In decision-based attack scenarios, attackers receive only the output label from the model after querying it. In response to the noncontinuous nature of the objective function's landscape, researchers have developed a variety of optimization problems~\cite{cheng:2019:query}. For instance, \citet{ilyas:2018:blackbox} propose using a surrogate for the objective function, while \citet{cheng:2020:signopt} and \citet{aithal:2022:boundary} innovate approaches based on geometric concepts. Furthermore, \citet{chen:2020:hop} tackle the original optimization problem by utilizing the gradient's sign. The principles outlined in our theoretical examination of score-based attacks are also applicable to decision-based attacks since both employ similar black-box optimization techniques.

\subsection{Adversarial Purification}

\begin{table*}[t]
    \centering
    \caption{\small List of heuristic transformations and SOTA purification models. Randomness is introduced in DISCO~\citep{ho:2022:disco} by using an ensemble of DISCO models to generate features for random coordinate querying, which is of high computational cost.\label{tab:purifiers}}
    \vspace{0.25\baselineskip}
    \begin{tabular}{lccc}
        \toprule
        Method                                     & Randomness            & Type      & Inference Cost \\
        \midrule
        Bit Reduction~\citep{xu:2018:feature}      & \ding{55}             & Heuristic & Low            \\
        Local Smoothing~\citep{xu:2018:feature}    & \ding{55}             & Heuristic & Low            \\
        JPEG Compression~\citep{raff:2019:barrage} & \ding{55}             & Heuristic & Low            \\
        Random Noise~\citep{qin:2021:random}       & \ding{51}             & Heuristic & Low            \\
        Score-based Model~\citep{yoon:2021:adv}    & \ding{51}             & Neural    & High           \\
        DDPM~\citep{nie:2022:diffusion}            & \ding{51}             & Neural    & High           \\
        DISCO~\citep{ho:2022:disco}                & \ding{55} / \ding{51} & Neural    & Median / High  \\
        \textbf{PuriDefense (Ours)}                & \ding{51}             & Neural    & Median         \\
        \bottomrule
    \end{tabular}
    \vspace{-1\baselineskip}
\end{table*}

Adversarial purification has recently emerged as a central wave of defense against adversarial attacks, which aims to remove or disturb the adversarial perturbations via \emph{heuristic transformations} and \emph{purification models}. We have provided a list of widely heuristic transformations and SOTA purification models in \tabref{tab:purifiers}.

\bheading{Heuristic Transformations.}
Heuristic transformations are unaware of the adversarial perturbations and only aim to disturb the adversarial perturbations by shrinking the image space~\citep{xu:2018:feature} or deviating the gradient estimation using random noises~\citep{qin:2021:random}.

\bheading{Purification Models.}
Advanced models for purification have been developed to eliminate adversarial perturbations. Notably, the score-based generative model~\citep{yoon:2021:adv} and DDPM~\citep{nie:2022:diffusion}, and local implicit purification models like DISCO~\citep{ho:2022:disco}. However, of these, only the local implicit purification demonstrates  a moderate inference cost necessary for practical, real-world applications.

With defense mechanisms deployed as pre-processors in the MLaaS system as shown in \Figref{fig:ml_system}, the attackers need to break the whole MLaaS pipeline to mount a successful attack. Although randomness has been identified as crucial in enhancing system robustness~\citep{raff:2019:barrage,sitawarin:2022:demystifying}, the naive implementation of randomness through the ensemble of multiple purifiers, such as DISCO, incurs a linear increase in the inference cost.

\section{Randomized Local Implicit Purification\label{sec:method}}

\subsection{Our Motivation}
While purification models can process adversarial images, our research, detailed in our theoretical analysis (\secref{subsec:theory}) and substantiated in the experiments (\secref{subsec:overall}), indicates that \emph{a single deterministic purifier} is insufficient for enhancing the robustness of a system against query-based attacks and may inadvertently introduce new vulnerabilities. While a direct ensemble of purifiers may seem effective in principle, the consequent linear increase in inference cost proportional to the number of purifiers renders the approach impractical for real-world deployment.

In response to this challenge, we propose a novel random patch-wise purification algorithm, PuriDefense, that capitalizes on a pool of purifiers to counter query-based attacks efficiently. PuriDefense employs multiple end-to-end purification models that utilize a local implicit function to process input images at any scale. Our theoretical findings demonstrate that the enhanced robustness of PuriDefense scales with the number of purifiers used. Most importantly, it maintains a consistent inference cost regardless of the number of purifiers, thereby offering a viable solution for deployment in practical settings.

\begin{figure*}[t]
    \centering
    \includegraphics[width=0.66\textwidth]{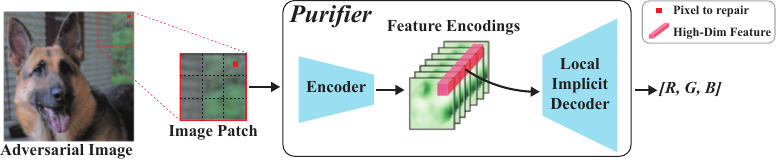}
    \vspace{-0.75\baselineskip}
    \caption{\small An illustration of repairing a pixel within an image patch with our end-to-end purification model.  The encoder diffuses nearby information of the pixel into its high-dimensional feature. Then the decoder reconstructs its RGB value with respect to this feature information. Note that the inference of pixels within one image patch can be performed in parallel.}
    \vspace{-1.\baselineskip}
    \label{fig:purifier}
\end{figure*}

\subsection{Image Purification via Local Implicit Function}
Consider a purification model $\vm(\vx): \mathcal{X} \to \mathcal{X}$ designed to project adversarial images back onto the natural images manifold. When attackers an craft adversarial example $\vx'$ from the original image $\vx$, randomly drawn from the natural image distribution $\mathcal{D}$. The purification model $\vm(\vx)$ can be refined by minimizing the following loss function:
\begin{equation}\label{eq:train_loss}
    \mathcal{L} = \mathbb{E}_{\mathcal{D}}\|\vx - \vm(\vx')\|_p + \lambda\mathbb{E}_{\mathcal{D}}\|\vx -\vm(\vx)\|_p,
\end{equation}
where the parameter $\lambda$ balances the two components of the loss. A larger $\lambda$ signifies a greater emphasis on fidelity to unaltered images. In practice, the second term is often disregarded ($\lambda=0$). The $\ell_1$-norm ($p=1$) is commonly utilized for its pixel-level accuracy in quantifying the discrepancy between the original and the purified image.

\bheading{Local Implicit Purification Model.}
To project the images suffering from adversarial perturbations onto the natural image manifold, we leverage a local implicit function designed to reconstruct the image area surrounding the perturbed pixels. Our model adopts an encoder-decoder framework. In this process, image patches are first encoded into a high-dimensional feature space. Subsequently, the decoder harnesses this space to restore the original RGB values of the pixels. This reconstruction is performed concurrently for all pixels within a patch, as depicted in \Figref{fig:purifier}, which also details the architecture of our model.

\bheading{Efficient Design.}
In contrast to the initial strategy of utilizing a local implicit function to defend against white-box attacks~\citep{ho:2022:disco}, our design removes the multi-scale support by dispensing with positional encoding and local ensemble inference. This simplification yields a fourfold acceleration in inference time at the implementation level. A comprehensive elucidation of this accelerated implementation is provided in Appendix~\ref{app:structure}.

\begin{figure}[t]
    \centering
    \includegraphics[width=0.45\textwidth]{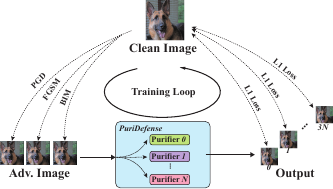}
    \vspace{-0.5\baselineskip}
    \caption{\small The training process of PuriDefense. Firstly, different adversarial images are generated using white-box attack algorithms. Then, we train every purifier inside PuriDefense with these adversarial images and original images under $\ell_1$ loss.\label{fig:training}}
    \vspace{-1.\baselineskip}
\end{figure}

\subsection{Training Purifiers in PuriDefense}
In PuriDefense, our strategy entails implementing a series of purification models to establish a varied ensemble for randomized patch-wise purification. Consequently, we have developed training protocols that enable the simultaneous training of diverse models, each with a distinct architecture. The procedural flow of our training methodology is delineated in Figure~\ref{fig:training}.

During the training process, adversarial samples generated by applying three advanced white-box attack techniques-PGD~\citep{madr:2018:towards}, FGSM~\citep{goodfellow:2015:explaining}, and BIM~\citep{kurakin:2017:adversarial_in}-to non-protected models. The purifiers then cleanse the perturbed inputs. Next, the purified samples and original images are used to calculate the training loss as per Equation~(\ref{eq:train_loss}). A back-propagation algorithm subsequently optimizes all the purifiers in PuriDefense. Additional information like regarding the non-defended models and the purifiers' architecture can be found in Appendix~\ref{app:diversity}.

\subsection{Random Patch-wise Purification.}
Random patch-wise purification constitutes the core of our defensive approach, designed to introduce randomness into the purification process, which is critical for defending against query-based attacks. It maintains the computational cost of employing multiple purifiers at a level that is equivalent to that of a single model, which is crucial for real-world deployment, as opposed to the current approach (ensemble) that escalates the inference cost proportionally to the number of purifiers employed~\citep{ho:2022:disco}.

Our method encodes image patches using randomly selected purifiers from the purifier pool and subsequently merges their outputs to construct the final feature representation, in contrast to previous techniques that necessitate full-image encoding followed by random feature selection.  As depicted in Figure \ref{fig:two_practice}, this approach—although employing potentially deterministic purifiers—ensures randomness by the unpredictable selection of both the purifiers for performing purifications. Such a strategy significantly augments the diversity of purifiers while maintaining reasonable inference cost.

\begin{figure*}[t]
    \centering
    \includegraphics[width=0.76\textwidth]{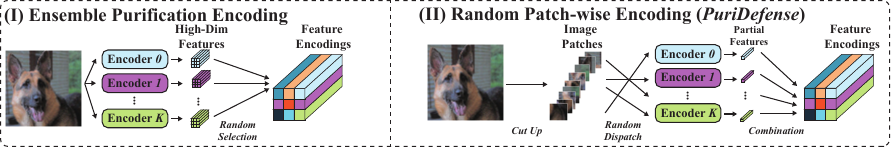}
    \vspace{-0.75\baselineskip}
    \caption{\small An illustration of the encoding process of ensemble method~\citep{ho:2022:disco} and our method. \textbf{Ensemble:} Ensemble method first encodes the image into multiple high-dimension features and then randomly combines their pieces to form the final feature representation. \textbf{Random Patch-wise:} Our method split the images into patches, forward them to randomly selected encoders, and use the output combination as the final feature representation. \label{fig:two_practice}}
    \vspace{-1.\baselineskip}
\end{figure*}

\subsection{Theoretical Analysis for Gradient-based Attacks} \label{subsec:theory}
Assume we have $K+1$ purifiers $\left\{\vm_0, \ldots, \vm_{K}\right\}$, the output of the new black-box system containing the $i$-th purifier is defined as $F^{(i)}(\vx) = f(\vm_i(\vx))$. Without loss of generality, we now perform analysis on breaking the system of the purifier $\vm_0$, denoted as $F(\vx) = f(\vm_0(\vx))$. Our following analysis utilizes the $\ell_2$-norm as the distance metric, which is the most commonly used norm for measuring the distance between two images.

Suppose the index of two independently drawn purifiers in our defense are $k_1$ and $k_2$, the attacker approximate the gradient of the function $F(\vx)$ with the following estimator:
\begin{equation}
    G_{\mu, K} = \frac{f(\vm_{k_1}(\vx+\mu \rvu)) - f(\vm_{k_2}(\vx))}{\mu}\rvu,
\end{equation}
where $\rvu$ is a standard Gaussian noise vector. The above gradient estimator provides an \emph{unbiased} estimation of the gradient of the function:
\begin{equation}
    F_{\mu, K}(\vx) = \frac{1}{K+1}\sum_{k=0}^{K}f_\mu(\vm_k(\vx)),
\end{equation}
where $f_\mu$ is the gaussian smoothing function of $f$. The detailed definition of the gaussian smoothing function is included in Appendix~\ref{app:definitions}. The convergence of the black-box attack is moving towards an \emph{averaged} optimal point of the functions $F^{(i)}$ formed with different purifiers $\vm_i$.

\bheading{Assumptions.}
For the original function $f(\vx)$, we have Assumption~\ref{assump:1} and Assumption~\ref{assump:2}. With regards to the purifiers, we assume each output dimension (every pixel in one channel) of their output also has the property in Assumption~\ref{assump:1} and Assumption~\ref{assump:2}. Then, we denote $L_0(\vm) = \max_i L_0(m_i)$ and $L_1(\vm) = \max_i L_1(m_i)$, where $m_i$ is the $i$-th dimension of the output of the purifier $\vm$.

\begin{MyAssump}
    \label{assump:1}
    $f(\vx)$ is Lipschitz-continuous, \ie, $|f(\vy)-f(\vx)|\leq L_0 (f)\|\vy-\vx\|$.
\end{MyAssump}
\begin{MyAssump}
    \label{assump:2}
    $f(\vx)$ is continuous and differentiable, and $\nabla f(\vx)$ is Lipschitz-continuous, \ie, $|\nabla f(\vy)-\nabla f(\vx)|\leq L_1 (f)\|\vy-\vx\|$.
\end{MyAssump}

Furthermore, we bound the diversity of the purifiers using the following property:
\begin{equation}
    \|\vm_i(\vx) - \vm_j(\vx)\| < \nu, \quad \forall i, j \in \left\{0, \ldots, K-1\right\}
\end{equation}
We cannot directly measure $\nu$, but we intuitively associate it with the number of purifiers. \textbf{The larger the number of purifiers, the larger $\nu$ is.}

\bheading{Notations.}
We denote the sequence of standard Gaussian noises used to approximate the gradient as $\rmU_t=\left\{\rvu_0,\ldots, \rvu_t\right\}$, with $t$ to be the update step. The purifier index sequence is denoted as $\rvk_t=\left\{\rk_0,\ldots, \rk_t\right\}$. The generated query sequence is denoted as $\left\{\vx_0, \vx_1,\dots, \vx_Q\right\}$. $d=|\mathcal{X}|$ as the input dimension.

With the above definitions and assumptions, we have Theorem~\ref{thm:1} for the convergence of the gradient-based attacks. The detailed proof is included in Appendix~\ref{app:proof1}.

\begin{MyTheo}
    \label{thm:1}
    Under Assumption 1, for any $Q\geq 0$, consider a sequence $\left\{\vx_t\right\}_{t=0}^Q$ generated using the update rule of gradient-based score-based attacks, with constant step size, \ie, $\eta = \sqrt{\frac{2R\epsilon }{(Q+1)L_0(f)^3d^2}}\cdot\sqrt{\frac{1}{L_0(\vm_0)\gamma(\vm_0, \nu)}}$, with $\gamma(\vm_0, \nu) = \frac{4\nu^2}{\mu^2} + \frac{4\nu}{\mu}L_0(\vm_0)d^{\frac{1}{2}}+L_0(\vm_0)^2d$. Then, the squared norm of gradient is bounded by:
    \begin{equation}
        \begin{aligned}
            \frac{1}{Q+1}\sum_{t=0}^Q \mathbb{E}_{\rmU_t, \rk_t}[\|\nabla F_{\mu, K} (\vx_t)\|^2] \\
            \leq \sqrt{\frac{2L_0(f)^5Rd^2}{(Q+1)\epsilon}}\cdot \sqrt{\gamma(\vm_0, \nu)L_0(\vm_0)^3}
        \end{aligned}
    \end{equation}
    The lower bound for the expected number of queries to bound the expected squared norm of the gradient of function $F_{\mu, K}$ of the order $\delta$ is
    \begin{equation}
        O(\frac{L_0(f)^5Rd^2}{\epsilon\delta^2}\gamma(\vm_0, \nu)L_0(\vm_0)^3)
    \end{equation}
\end{MyTheo}

\bheading{Single Deterministic Purifier.}
Setting $\nu$ to 0, we have $\gamma(\vm_0, 0)L_0(\vm_0)^2 = L_0(\vm_0)^5$, which is the only introduced term compared to the original convergence rate~\citep{nesterov:2017:random} towards $f(\vx)$. Meanwhile, the new convergence point becomes $F^*_\mu(\vx)$. We have the following conclusion for the convergence of the attack:
\begin{itemize}
    \item \textbf{Influence of $L_0(\vm_0)$}: For input transformations that \emph{shrink} the image space, since their $L_0(\vm_0) < 1$, they always allow a \emph{faster} rate of attack's convergence. For neural network purifiers, the presence of this term means their vulnerabilities is introduced into the black-box system, making it hard to quantify the robustness of the system.
    \item \textbf{Optimal point $F^*_\mu(\vx)$}: By using a deterministic transformations, the optimal point of the attack is changed from $f^*$ to $F^*_\mu(\vx)$. If we can inversely find an adversarial image $\vx^* = \vm(\vx^*)$, the robustness of the system is not improved at all. \emph{No current work can theoretically eliminate this issue.} This may open up a new direction for future research.
\end{itemize}

\begin{tcolorbox}
    [%
        enhanced,
        breakable,
        boxrule=0.1pt,
        left=1pt,
        right=1pt,
        bottom=1pt,
        top=1pt
    ]{
        \textbf{Research implications.}
        \small
        From the above analysis, we can see that a single deterministic purifier may \emph{1) accelerate} the convergence of the attack, and \emph{2) cannot protect} the adversarial point from being exploited.
    }
\end{tcolorbox}


\bheading{Pool of Deterministic Purifiers.}
The introduced term $\gamma(\vm_0, \nu)L_0(\vm_0)^2$ increase quadratically with $\nu$. This along with our intuition mentioned above suggests that \emph{the robustness of the system increases with the number of purifiers.} While adversarial optimal points persist, the presence of multiple optimal points under different purifiers serve as the \emph{first} trial to enhance the robustness of all purification-based methods.

To validate our theoretical analysis, we conduct experiments on the test set of the CIFAR-10 dataset~\citep{CIFAR10} with a weak non-defended ResNet-18 model~\citep{resnet18_cifar} as the classifier. Other general settings are the same as used in our later experiments. We use the Square Attack~\citep{andriushchenko:2020:square} as the attack algorithm. The convergence of the attack against our model and other input transformations is shown in \Figref{fig:convergence}. We can see a clear acceleration of the convergence of the attack with the introduction of transformations that \emph{shrink} the image space and powerful deterministic models (DISCO) fails to improve the robustness of the system. Another validation of our theoretical analysis is shown in \Figref{fig:nu_convergence} for proving the robustness of the system increases with the number of purifiers (associated with $\nu$).

\begin{figure}[t]
    \centering
    \includegraphics[width=0.42\textwidth]{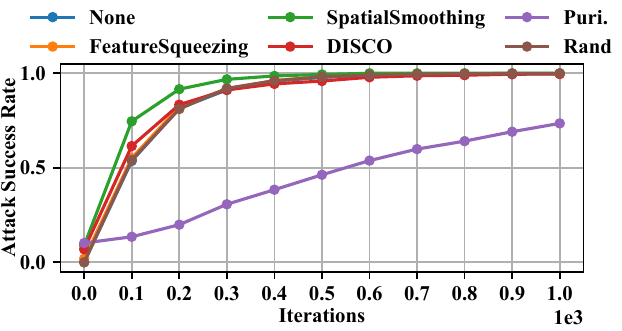}
    \vspace{-0.75\baselineskip}
    \caption{\small The convergence of the Square Attack on CIFAR-10 using different heuristic transformations and purifiers.\label{fig:convergence}}
    \vspace{-1.\baselineskip}
\end{figure}

\subsection{Theoretical Analysis for Gradient-free Attacks}
The heuristic direction of random search becomes:
\begin{equation}
    \label{eq:search}
    H_K(\vx) = f(\vm_{k_1}(\vx+\mu \ru)) - f(\vm_{k_2}(\vx+\mu \ru)).
\end{equation}
\begin{MyTheo}
    \label{thm:2}
    Under Assumption~\ref{assump:1}, using the update in Equation~(\ref{eq:search}),
    \begin{equation}
        P(Sign(H(\vx)) \neq Sign(H_K(\vx))) \leq \frac{2\nu L_0(f)}{|H(x)|}
    \end{equation}
\end{MyTheo}

A similar increase in the robustness as Theorem~\ref{thm:1} can be observed with the increase of $\nu$. The detailed proof is included in Appendix~\ref{app:proof2}. This ensures the robustness of our defense against gradient-free attacks.

%% file: 004-evaluation.tex
\section{Evaluation\label{sec:eva}}

\subsection{Experiment Settings}
\bheading{Datasets and Classification Models.}
For a comprehensive evaluation of PuriDense, we employ two benchmark datasets for testing adversarial attacks: CIFAR-10~\citep{CIFAR10} and ImageNet~\citep{ImageNet}. Our evaluation is conducted on two balanced subsets, which contain 1,000 and 2,000 test images randomly sampled from the CIFAR-10 test set and ImageNet validation set, respectively. These subsets are evenly distributed across the 10 classes of CIFAR-10 and 200 randomly chosen classes from ImageNet. In terms of classification models, we selected models from RobustBench~\citep{croce:2021:robustbench}. The details of both the standardly trained and adversarially trained models are described in Appendix~\ref{app:model}.

\begin{figure}[t]
  \centering
  \includegraphics[width=0.42\textwidth]{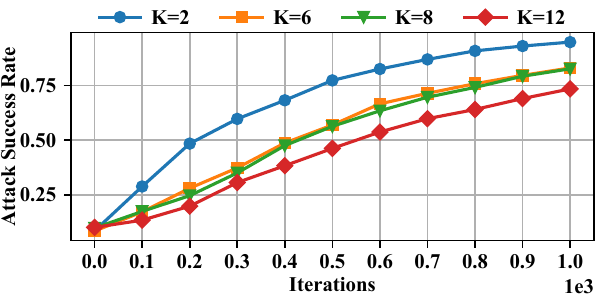}
  \vspace{-1\baselineskip}
  \caption{\small The convergence of the Square Attack on CIFAR-10 using with different numbers of purifiers used.\label{fig:nu_convergence}}
  \vspace{-1.25\baselineskip}
\end{figure}

\bheading{Attack and Defense Methods.}
We consider 5 query-based attacks for evaluation: NES~\citep{ilyas:2018:blackbox}, SimBA~\citep{guo:2019:simba}, Square~\citep{andriushchenko:2020:square}, Boundary~\citep{brendel:2018:decision}, and HopSkipJump~\citep{chen:2020:hop}. Comprehensive descriptions and configurations of each attack can be found in Appendix~\ref{app:attack}. The perturbation budget of $\ell_\infty$ attacks is set to 8/255 for CIFAR-10 and 4/255 for ImageNet. For $\ell_2$ attacks, the perturbation budget is set to 1.0 for CIFAR-10 and 5.0 for ImageNet. For defense mechanism, adversarially trained models from \citep{gowal:2020:uncovering} are used as a strong robust baseline. Moreover, we include a deterministic purification model DISCO~\cite{ho:2022:disco} and spatial smoothing~\citep{xu:2018:feature} for direct comparison. Finally, widely used random noise defense~\citep{qin:2021:random} serve as a baseline for introducing randomness. The detailed settings of each defense method are described in Appendix~\ref{app:defense}. We report the robust accuracy of each defense method against each attack with 200/2500 queries for both CIFAR-10 and ImageNet,  reflecting models' performance under mild and extreme query-based attacks.

\begin{table*}[h]
  \centering
  \caption{Evaluation results of PuriDefense and other defense methods on CIFAR-10 under 5 query-based attacks. The robust accuracy under 200/2500 queries is reported. The best defense mechanism under 2500 queries are highlighted in bold and marked with gray.\label{tab:overall}}
  \vspace{0.25\baselineskip}
  \resizebox{0.98\textwidth}{!}{
    \begin{tabular}{|c|l|c|c|c|c|c|c|}
      \hline
      Datasets                                                                       & Methods                          & Acc.                                                      & NES($\ell_\infty$)                                             & SimBA($\ell_2$)                                                & Square($\ell_\infty$)                                          & Boundary($\ell_2$)                                             & HopSkipJump($\ell_\infty$)                                     \\
      \hline
      \multirow{6}{*}{\begin{tabular}{c} CIFAR-10 \\  (WideResNet-28)\end{tabular} } & None                             & 94.8                                                      & 83.4/11.9                                                      & 49.1/3.0                                                       & 26.6/0.9                                                       & 88.4/60.0                                                      & 72.0/72.6                                                      \\
      \cline{2-8}
                                                                                     & AT~\citep{gowal:2020:uncovering} & 85.5                                                      & 83.8/78.8                                                      & 83.0/74.9                                                      & 77.5/67.3                                                      & \cellcolor{light-gray}\cellcolor{light-gray}\textbf{84.7/84.2} & 85.3/84.0                                                      \\
                                                                                     & FS~\citep{xu:2018:feature}       & 76.4                                                      & 50.7/7.6                                                       & 42.7/5.2                                                       & 5.7/0.0                                                        & 70.7/39.8                                                      & 64.6/62.7                                                      \\
                                                                                     & IR~\citep{qin:2021:random}       & 77.1                                                      & 74.7/71.1                                                      & 71.5/66.0                                                      & 64.1/60.5                                                      & 74.4/76.4                                                      & 71.0/74.1                                                      \\
                                                                                     & DISCO~\citep{ho:2022:disco}      & \cellcolor{light-gray}\cellcolor{light-gray}\textbf{86.3} & 83.0/34.7                                                      & 77.7/15.7                                                      & 21.0/2.1                                                       & 84.1/66.0                                                      & 81.0/81.7                                                      \\
                                                                                     & \textbf{PuriDefense (Ours)}      & 84.1                                                      & \cellcolor{light-gray}\cellcolor{light-gray}\textbf{84.2/81.8} & 81.1/74.4                                                      & 77.3/67.8                                                      & 82.9/82.9                                                      & \cellcolor{light-gray}\cellcolor{light-gray}\textbf{82.5/84.8} \\
                                                                                     & \textbf{PuriDefense-AT (Ours)}   & 84.6                                                      & 83.7/81.1                                                      & \cellcolor{light-gray}\cellcolor{light-gray}\textbf{83.3/80.6} & \cellcolor{light-gray}\cellcolor{light-gray}\textbf{81.5/78.7} & 84.4/84.1                                                      & 84.3/84.1                                                      \\
      \hline
    \end{tabular}
  }
  \vspace{-0.5\baselineskip}
\end{table*}

\subsection{Overall Defense Performance}\label{subsec:overall}

Our numerical results on the effectiveness of the defense mechanisms are shown in Table~\ref{tab:overall}.

\bheading{Clean Accuracy.}
Empirical evaluations demonstrate that PuriDefense yields clean accuracy comparable to the standardly trained model. Notably, integrating PuriDefense with models trained adversarially enhances their clean accuracy at no additional cost. Further experiments offer an in-depth analysis of the variations in clean accuracy with the implementation of PuriDefense.

\bheading{Failure of Deterministic Purification.}
As presented in \tabref{tab:overall}, spatial smoothing accelerates the convergence of the attacks, and DISCO experiences a significant decrease in robust accuracy under 2500 queries when faced with a powerful attack (Square Attack), which is even lower than that of models without defenses. These outcomes reinforce our theoretical discourse, suggesting that deterministic transformations may inadvertently introduce additional vulnerabilities and expedite adversarial attacks. Consequently, incorporating randomness into purification processes is not only theoretically grounded but also empirically validated.

\bheading{Effectiveness of PuriDefense.}
PuriDefense consistently attains the highest robust accuracy against various attacks on the CIFAR-10 dataset under extreme query-based attacks (2500 queries). When integrated with standardly trained model, the system reaches a robust accuracy comparable to SOTA adversarially trained models. Furthermore, PuriDefense, when combined with adversarially trained models, sets new benchmarks in robust accuracy. Its efficacy against a variety of query-based attacks demonstrates PuriDefense's versatility as a universal defense mechanism effective against both $\ell_\infty$ and $\ell_2$ attacks.

\begin{figure}[t]
  \centering
  \includegraphics[width=0.42\textwidth]{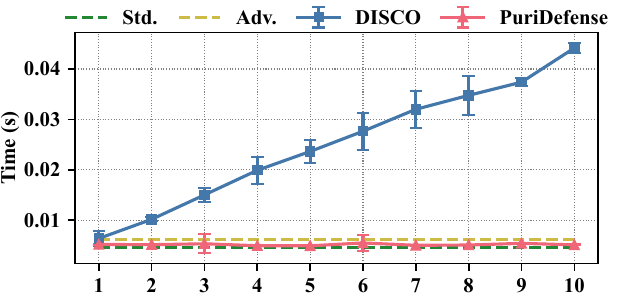}
  \vspace{-0.8\baselineskip}
  \caption{\small The inference speed of DISCO, PuriDefense, along with the standalone classifier with or without adversarial training on CIFAR-10 dataset.\label{fig:cifar10_latency}}
  \vspace{-1\baselineskip}
\end{figure}

\subsection{Inference Speedup}
Our mechanism achieves a constant inference cost, which we verify in this section. The inference speed of DISCO and PuriDefense was tested on a workstation equipped with a single NVIDIA RTX 4090 GPU. We set the batch size to 1 and varied the number of purifiers from 1 to 10. To minimize the influence of data transmission delay, we measured the time of the last 900 inferences out of a total of 1000 iterations. To ensure fair comparison, both DISCO and PuriDefense utilized identical encoders and decoders; therefore, differences in inference stem solely from their respective methods. We employed the same models, standardized and adversarially trained, as established in \secref{sec:eva} for baseline comparison.

\bheading{Results.}
For CIFAR-10, the outcomes are presented in Figure~\ref{fig:cifar10_latency}. Additional results for ImageNet are provided in Appendix~\ref{app:infer_speed}, corroborating the consistency of the findings. Unlike our method, which sustains a steady inference cost regardless of the number of purifiers, DISCO's cost escalates almost linearly. PuriDefense maintains its cost at the baseline classifier level, demonstrating its advantage over prevailing diffusion-based purification models.

\subsection{Influence on Clean Accuracy}
To further assess the impact of PuriDefense on clean accuracy, we evaluate the performance decline for both CIFAR-10 and ImageNet datasets when utilizing all the standalone purification models implemented in DISCO and PuriDefense. Results in \Figref{fig:clean_acc} in Appendix~\ref{app:clean} show that purification models illustrate comparable clean accuracy on low-resolution datasets, \ie CIFAR-10, and achieves a higher clean accuracy on high-resolution datasets, \ie ImageNet. We attribute this to the fact that natural image manifold exists in high-resolution datasets.

To test the influence of the number of the image patches on clean accuracy, we vary the number of patches ranging in  $\left\{ 1\times 1, 3\times 3, 5\times 5\right\}$ using ImageNet dataset. The results in \Figref{fig:clean_acc_matt} in Appendix~\ref{app:clean} shows that increasing the number of patches does not significantly affect the clean accuracy, PuriDefense achieves comparable clean accuracy to the case without any defense mechanism.

%% file: 006-app.tex
\section{Background for General Defense\label{app:back}}
\bheading{More on General Defense for Query-based Attacks}
The defense for black-box query-based attacks remain relatively unexplored compared to the defense for white-box methods~\cite{qin:2021:random}. Under considerations of real-world constraints such as clean accuracy and inference speed, most training-time and testing-time defenses present significant limitations for deployment in real-world MLaaS systems.  For training-time defense, the aim is to improve the worst-case robustness of the models. Adversarial training (AT) has been considered as one of the fundamental practices of training time defense, where models are trained on augmented datasets with specially-crafted samples to ensure robustness of the models~\citep{madr:2018:towards}. Other training-time examples like gradient masking~\citep{florian:2018:ensemble} and defensive distillation~\cite{papernot:2016:distillation} are also proposed to improve the robustness of the models. Nonetheless, such methods are unsuitable for MLaaS systems because of the extensive training costs and potential for decreased accuracy on clean examples. With regards to testing-time defense, a prominent defense from white-box attacks, randomized smoothing, can ensure the robustness of the model within a certain confidence level~\citep{yang:2020:randomized}, a feature known as certified robustness. Another example for multiple inference to improve the robustness is called random self-ensemble~\cite{liu:2018:towards}. However, the inference speed of randomized smoothing is too slow to be deployed on real-world MLaaS systems. Other testing-time defenses tend towards randomization of the input or output. Rand noise defense proposed by \citet{qin:2021:random} leverages Gaussian noises as the input to the model to disturb the gradient estimation. Yet, the defense is ineffective against strong attack methods and hurts the clean accuracy. The output-based defense, like confidence poisoning~\citep{chen:2022:aaa} influences the examples on the classification boundary and cannot defend against the decision-based attacks.

\section{Search Techniques for black-box attacks\label{app:search}}
\bheading{Projected Gradient Descent}
A common approach of performing adversarial attacks (often white-box) is to leverage projected gradient descent algorithm~\citep{carlini:2017:towards}:
\begin{equation}
    \label{eq:pgd}
    \vx_{t+1} = Proj_{\mathcal{N}_R(\vx)}(\vx - \eta_t g(\vx)_t).
\end{equation}

\bheading{Gradient Estimation}
While there can be various gradient estimators, we consider the following gradient estimator in our theoretical analysis:
\begin{equation}
    \label{eq:gradient_0}
    g(\vx) = \frac{f(\vx+\mu \rvu) - f(\vx)}{\mu}\rvu.
\end{equation}

\bheading{Heuristic Search}
For heuristic search, the main issue is to determine the search direction. One commonly used search direction can be:
\begin{equation}
    s(\vx) = \mathbb{I}(h(\vx)<0)\cdot\mu\rvu,\quad\text{where} \quad h(\vx) = f(\vx+\mu \rvu) - f(\vx),
\end{equation}
where $\mathbb{I}$ is the indicator function. The search direction is determined by the sign of the objective function. If the objective function is negative, the search direction is the gradient direction. Otherwise, the search direction is the opposite of the gradient direction. The corresponding updating direction will be Equation~(\ref{eq:pgd}) with $- \eta_t g(\vx)_t$ replaced by $s(\vx_t)$.

\section{Details for Experiments}
\subsection{Comparison of the Inference Speed on ImageNet Dataset\label{app:infer_speed}}
\begin{figure}[h]
    \centering
    \includegraphics[width=0.45\textwidth]{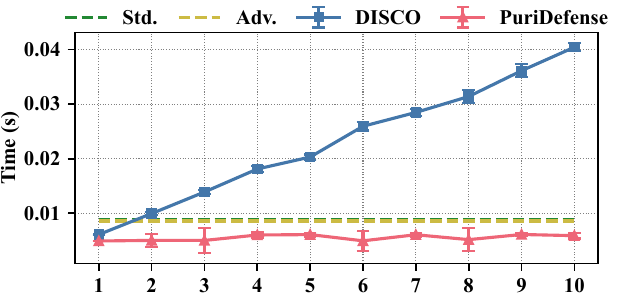}
    \vspace{-0.75\baselineskip}
    \caption{The inference speed of DISCO and PuriDefense on ImageNet dataset.\label{fig:imagenet_latency}}
    \vspace{-1.\baselineskip}
\end{figure}
\subsection{Details of the Attacks\label{app:attack}}
We utilize 5 SOTA query-based attacks for evaluation: NES~\citep{ilyas:2018:blackbox}, SimBA~\citep{guo:2019:simba}, Square~\citep{andriushchenko:2020:square}, Boundary~\citep{brendel:2018:decision}, and HopSkipJump~\citep{chen:2020:hop}. The category of them is listed below in \tabref{tab:cate}.
\begin{table}[h]
    \centering
    \caption{The category of the attacks along with the techniques they use.\label{tab:cate}}
    \begin{tabular}{ccc}
        \toprule
                       & Gradient Estimation & Random Search \\
        \midrule
        Score-based    & NES                 & Square, SimBA \\
        Decision-based & Boundary            & HopSkipJump   \\
        \bottomrule
    \end{tabular}
\end{table}

\bheading{Implementation}
For Boundary Attack and HopSkipJump Attack, we adopt the implementation from Foolbox~\citep{jonas:2020:foolbox}. For Square Attack and SimBA, we use the implementation from ART library~\citep{nicolae:2018:art}. For NES, we implement it under the framework of Foolbox.

\bheading{Hyperparameters}
The hyperparameters used for the attacks are listed below for full reproducibility.
\begin{minipage}{\textwidth}

    \begin{minipage}[t]{\textwidth}
        \makeatletter\def\@captype{table}
        \centering
        \caption{The hyperparameters used for NES.}
        \begin{tabular}{lcc}
            \toprule
                                                       & CIFAR-10 & ImageNet \\
            \midrule
            $\eta$ (learning rate)                     & 0.01     & 0.0005   \\
            $q$ (number of points used for estimation) & 100      & 100      \\
            \bottomrule
        \end{tabular}
    \end{minipage}
    \begin{minipage}[t]{0.47\textwidth}
        \makeatletter\def\@captype{table}
        \centering
        \caption{The hyperparameters used for SimBA.}
        \resizebox{!}{1.5\baselineskip}{
            \begin{tabular}{lcc}
                \toprule
                                   & CIFAR-10 & ImageNet \\
                \midrule
                $\eta$ (step size) & 0.2      & 0.2      \\
                \bottomrule
            \end{tabular}
        }
    \end{minipage}
    \begin{minipage}[t]{0.52\textwidth}
        \makeatletter\def\@captype{table}
        \centering
        \caption{The hyperparameters used for Square.}
        \resizebox{!}{1.5\baselineskip}{
            \begin{tabular}{lcc}
                \toprule
                                                  & CIFAR-10       & ImageNet       \\
                \midrule
                $\mu$ (Fraction of Pixel Changed) & $0.05\sim 0.5$ & $0.05\sim 0.5$ \\
                \bottomrule
            \end{tabular}
        }
    \end{minipage}
    \begin{minipage}[t]{0.48\textwidth}
        \makeatletter\def\@captype{table}
        \centering
        \caption{The hyperparameters used for Boundary Attack.}
        \resizebox{!}{2\baselineskip}{
            \begin{tabular}{lcc}
                \toprule
                                                  & CIFAR-10 & ImageNet \\
                \midrule
                $\eta_{sph}$ (Spherical Step)     & 0.01     & 0.01     \\
                $\eta_{src}$ (Source Step)        & 0.01     & 0.01     \\
                $\eta_{c}$ (Source Step Converge) & 1E-7     & 1E-7     \\
                $\eta_{a}$ (Step Adaptation)      & 1.5      & 1.5      \\
                \bottomrule
            \end{tabular}
        }
    \end{minipage}
    \begin{minipage}[t]{0.48\textwidth}
        \makeatletter\def\@captype{table}
        \centering
        \caption{The hyperparameters used for HopSkipJump Attack.}
        \resizebox{!}{2\baselineskip}{
            \begin{tabular}{lcc}
                \toprule
                                               & CIFAR-10 & ImageNet \\
                \midrule
                $n$ (number for estimation)    & 100      & 100      \\
                $\gamma$ (Step Control Factor) & 1        & 1        \\
                \bottomrule
            \end{tabular}
        }
    \end{minipage}
\end{minipage}

\subsection{Detailed Information for the Defense\label{app:defense}}
We compare our algorithm with three types of baseline defense. For random noise defense, we use a Gaussian noise with $\sigma=0.041$ as the input to the classifier. For the spatial smoothing transformations, we set the size of the kernel filter to be 3. For DISCO model, we implement a naive version without randomness using pre-trained models from the official implementation. We pick the pre-trained model under PGD attack~\citep{madr:2018:towards} as the core local implicit model for DISCO.

\subsection{Testing models.}\label{app:model}
We use the WideResNet-28-10~\citep{sergy:2016:wideres} achieving a 94.78\% accuracy rate for CIFAR-10, and ResNet-50~\citep{kaiming:2016:resnet} with 76.52\% accuracy for ImageNet for the standardly trained models. For models trained adversarially, we use the WideResNet-28-10 model with 89.48\% adversarial accuracy trained by \citet{gowal:2020:uncovering} for CIFAR-10 and ResNet-50 model with 64.02\% trained by \citet{salman:2020:do} for ImageNet.

\section{Details for PuriDefense\label{app:setting}}
\subsection{Efficient Structure for Inference Speed Up\label{app:structure}}
Upon examining the implementation of the local implicit function as outlined by \citep{chen:2021:learning}, it became apparent that the local ensemble mechanism geared toward enhancing performance in super-resolution tasks is unnecessary for the adversarial purification process. Refer to \Figref{fig:local_ensemble} for a conceptual depiction of the previously utilized technique.

Given that image purification requires direct one-to-one pixel correspondence, the act of inferring the same pixel multiple times before averaging the outcomes is redundant. Consequently, discarding this mechanism simplifies the approach to using the local implicit function solely for image purification. This adjustment accelerates the inference speed of the original implementation of local implicit models by a factor of four.

\begin{figure}[h]
    \centering
    \includegraphics[width=0.25\textwidth]{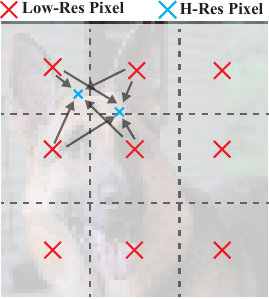}
    \vspace{-0.75\baselineskip}
    \caption{\small An illustration of the local ensemble mechanism in the local implicit function for multi-resolution support. High resolution pixels are predicted based on high-level features from nearby low resolution pixels. \label{fig:local_ensemble}}
    \vspace{-1.\baselineskip}
\end{figure}

\subsection{Training Diversified Purifiers\label{app:diversity}}
For improving the diversity of the purifiers, we consider the following influence factors in \tabref{tab:factor} and use their combinations to train 12 different purification models for each dataset.

\bheading{Non-defended Models.}
For CIFAR-10, we use a pre-trained ResNet-18 model~\citep{resnet18_cifar} for generating adversarial examples under white-box attacks. For ImageNet, we use a pre-trained ResNet-50 model~\citep{resnet50_image} for generating adversarial examples. We list these models along with their clean accuracy on the test set in \tabref{tab:non_defended}.

\bheading{Training Datasets.}
For CIFAR-10, we use the whole training set of CIFAR-10 to train the purifiers. While for datasets containing natural images, using only a subset of the training set can help the purification model learn the natural image. Thus, for ImageNet dataset, we randomly sampled $10$ images from each class of the original training set to form a new training set for training the purifiers.

\begin{table}[h]
    \centering
    \caption{\small The factors that are considered when training diversified purification models for PuriDefense.\label{tab:factor}}
    \vspace{0.25\baselineskip}
    \begin{tabular}{ll}
        \toprule
        Hyperparameter                     & Value                                   \\
        \midrule
        \multirow{3}{*}{Attack Type }      & FGSM~\citep{goodfellow:2015:explaining} \\
                                           & PGD~\citep{madr:2018:towards}           \\
                                           & BIM~\cite{kurakin:2017:adversarial}     \\
        \hline
        \multirow{2}{*}{Encoder Structure} & RCAN~\citep{zhang:2018:image}           \\
                                           & EDSR~\citep{lim:2017:edsr}              \\
        \hline
        Feature Depth                      & 32, 64                                  \\
        \bottomrule
    \end{tabular}
\end{table}

\begin{table}[h]
    \centering
    \caption{\small The non-defended models used to generate adversarial examples.\label{tab:non_defended}}
    \vspace{0.25\baselineskip}
    \begin{tabular}{lll}
        \toprule
        Model Structure                  & Dataset  & Clean Accuracy \\
        \midrule
        ResNet-18~\citep{resnet18_cifar} & CIFAR-10 & 94.8\%         \\
        \hline
        ResNet-50~\citep{resnet50_image} & ImageNet & 80.9\%         \\
        \bottomrule
    \end{tabular}
\end{table}


\section{Supplementary Materials for Theoretical Analysis}
\subsection{Important Definitions\label{app:definitions}}
\begin{MyDef}
    The Gaussian-Smoothing function corresponding to $f(\vx)$ with $\mu>0, \rvu\sim\mathcal{N}(\bm{0}, \bm{I})$ is
    \begin{equation}
        f_\mu(\vx) = \frac{1}{(2\pi)^{d/2}}\int f(\vx+\mu\rvu)e^{-\frac{\|\rvu\|^2}{2}}\mathrm{d}\rvu
    \end{equation}
\end{MyDef}
\subsection{Proof of Theorem 1\label{app:proof1}}
The essential lemmas are given as follows, the complete proofs are shown in \cite{nesterov:2017:random}.
\begin{MyAppLemma}
    \label{app_lemma:1}
    Let $f(\vx)$ be the Lipschitz-continuous function, $|f(\vy) - f(\vx)| \leq L_0(f)\|\vy - \vx\|$. Then
    \begin{equation*}
        L_1(f_\mu) = \frac{d^{\frac{1}{2}}}{\mu} L_0(f)
    \end{equation*}
\end{MyAppLemma}

We define the $p$-order moment of normal distribution as $M_p$. Then we have
\begin{MyAppLemma}
    \label{app_lemma:2}
    For $p\in [0, 2]$, we have
    \begin{equation*}
        M_p \leq d^{\frac{p}{2}}
    \end{equation*}
    If $p \geq 2$, the we have two-side bounds
    \begin{equation*}
        d^{\frac{p}{2}} \leq M_p \leq (p+d)^{\frac{p}{2}}
    \end{equation*}
\end{MyAppLemma}

\begin{MyAppLemma}
    \label{app_lemma:3}
    Let $f(\vx)$ be the Lipschitz-continuous function, $|f(\vy) - f(\vx)| \leq L_0(f)\|\vy - \vx\|$. And $\vm(\vx)$ is Lipschitz-continuous for every dimension. Then
    \begin{equation*}
        L_0(f\circ \vm) \leq L_0(f)L_0(\vm)
    \end{equation*}
    where $L_0(\vm)$ is defined as $L_0(\vm) = \max_i L_0(m_i)$.
\end{MyAppLemma}
\begin{proof}
    \begin{equation}
        \begin{aligned}
            |f(\vm(\vy)) - f(\vm(\vx))| & \leq L_0(f)\|\vm(\vy) - \vm(\vx)\|                       \\
                                        & = L_0(f)\sqrt{\sum_{i=1}^dL_0(\vm_i)^2(\vy_i - \vx_i)^2} \\
                                        & \leq L_0(f)L_0(\vm)\|\vy - \vx\|
        \end{aligned}
    \end{equation}
\end{proof}



The following content is the proof for Theorem~\ref{thm:1}.

\begin{proof}
    According to the property of Lipschitz-continuous gradient,
    \begin{equation}
        F_{\mu, K} (\vx_{t+1}) \leq F_{\mu, K} (\vx_{t}) - \eta_t \langle \nabla F_{\mu, K} (\vx_t), G_{\mu, K} (x_t) \rangle + \frac{1}{2} \eta^2_t L_1 (F_{\mu, K}) \|G_{\mu, K} (\vx_t)\|^2
    \end{equation}
    The $G_{\mu,K}(\vx_t)$ can be decomposed as
    \begin{equation}
        \begin{aligned}
            G_{\mu, K} (\vx_t) & = \frac{f(\vm_{k_{t1}}(\vx+\mu \rvu)) - f(\vm_{k_{t2}}(\vx))}{\mu}\rvu_t                                            \\
                               & = \frac{f(\vm_{k_{t1}}(\vx+\mu \rvu)) - f(\vm_0(\vx+\mu \rvu)) + f(\vm_0(\vx+\mu \rvu)) - f(\vm_0(\vx))}{\mu}\rvu_t \\
                               & + \frac{f(\vm_0(\vx)) - f(\vm_{k_{t2}}(\vx))}{\mu}\rvu_t                                                            \\
        \end{aligned}
    \end{equation}
    The squared term $\|G_{\mu, K} (\vx_t)\|^2$ is bounded by
    \begin{equation}
        \begin{aligned}
            \|G_{\mu, K} (\vx_t)\|^2 & \leq \frac{4\nu^2}{\mu^2}L_0(f)^2\|\rvu_t\|^2 + \frac{4\nu}{\mu}L_0(F)L_0(f)\|\rvu_t\|^3  + L_0(F)^2\|\rvu_t\|^4 \\
        \end{aligned}
    \end{equation}
    Take the expectation over $\rvu_t$, $k_{t1}$, and $k_{t2}$, use Lemma~\ref{app_lemma:2}, we have
    \begin{equation}
        \begin{aligned}
            F_{\mu, K} (\vx_{t+1}) & \leq F_{\mu, K} (\vx_{t}) - \eta_t\|\nabla F_{\mu, K} (\vx_t)\|^2                                                                     \\
                                   & +\frac{1}{2}\eta_t^2L_1(F_{\mu,K})(\frac{4\nu^2}{\mu^2}L_0(f)^2d + \frac{4\nu}{\mu}L_0(F)L_0(f)(d+3)^{\frac{3}{2}} + L_0(F)^2(d+4)^2)
        \end{aligned}
    \end{equation}
    For $L_1(F_{\mu, K})$, we have:
    \begin{equation}
        L_1(F_{\mu, K}) = \frac{1}{K}\sum_{k=1}^K L_1(f_{\mu}(\vm_k)) \leq  \frac{L_0(F)d^{\frac{1}{2}}}{\mu}
    \end{equation}
    Use Lemma~\ref{app_lemma:1}, and the dimension $d$ is high, we have
    \begin{equation}
        \begin{aligned}
            F_{\mu,K}(\vx_{t+1}) & \leq F_{\mu, K} (\vx_{t}) - \eta_t\|\nabla F_{\mu, K} (\vx_t)\|^2                                                                                     \\
                                 & + \frac{1}{2}\eta_t^2 \frac{L_0(f)^3L_0(\vm_0)d^{\frac{3}{2}}}{\mu}(\frac{4\nu^2}{\mu^2} + \frac{4\nu}{\mu}L_0(\vm_0)d^{\frac{1}{2}} + L_0(\vm_0)^2d)
        \end{aligned}
    \end{equation}

    We take the expectation over $\rmU_t, \rk_t$.
    \begin{equation}
        \begin{aligned}
            \mathbb{E}_{\rmU_t, \rk_t}[F_{\mu, K} (\vx_{t+1})] & \leq \mathbb{E}_{\rmU_{t-1}, \rk_{t-1}}[F_{\mu, K} (\vx_{t})] - \eta_t\mathbb{E}_{\rmU_t, \rk_t}[\|\nabla F_{\mu, K} (\vx_t)\|^2]                     \\
                                                               & + \frac{1}{2}\eta_t^2 \frac{L_0(f)^3L_0(\vm_0)d^{\frac{3}{2}}}{\mu}(\frac{4\nu^2}{\mu^2} + \frac{4\nu}{\mu}L_0(\vm_0)d^{\frac{1}{2}} + L_0(\vm_0)^2d)
        \end{aligned}
    \end{equation}
    Now consider constant step size $\eta_t=\eta$, and sum over $t$ from $0$ to $Q$, we have
    \begin{equation}
        \begin{aligned}
            \frac{1}{Q+1}\sum_{t=0}^Q \mathbb{E}_{\rmU_t}[\|\nabla F_{\mu, K} (\vx_t)\|^2] & \leq \frac{1}{\eta}(\frac{F_{\mu, K} (\vx_{0}) - F_K^*}{Q+1})                                                                                     \\
                                                                                           & + \frac{1}{2}\eta \frac{L_0(f)^3L_0(\vm_0)d^{\frac{3}{2}}}{\mu}(\frac{4\nu^2}{\mu^2} + \frac{4\nu}{\mu}L_0(\vm_0)d^{\frac{1}{2}} + L_0(\vm_0)^2d)
        \end{aligned}
    \end{equation}
    Since the distance between the input variable should be bounded by $R$ and use Lipschitz-continuous, we have
    \begin{equation}
        \|F_{\mu, K}(\vx_0) - F^*_K\| \leq \frac{1}{K}L_0(f)\sum_{k=0}^K L_0(\vm_k) R \leq L_0(F) R
    \end{equation}
    Considering bounded $\mu \leq \hat{\mu} = \frac{\epsilon}{d^{\frac{1}{2}}L_0(F)}$ to ensure local Lipschitz-continuity, and set $\gamma(\vm_0, \nu) = \frac{4\nu^2}{\mu^2} + \frac{4\nu}{\mu}L_0(\vm_0)d^{\frac{1}{2}}+L_0(\vm_0)^2d$
    \begin{equation}
        \begin{aligned}
            \frac{1}{Q+1}\sum_{t=0}^Q \mathbb{E}_{\rmU_t, \rk_t}[\|\nabla F_\mu (\vx_t)\|^2] \leq \frac{1}{\eta}(\frac{L_0(F)R}{Q+1}) + \frac{1}
            {2}\eta \frac{L_0(f)^4L_0(\vm_0)^2}{\epsilon }d^2\gamma(\vm_0, \nu)
        \end{aligned}
    \end{equation}
    Minimize the right hand size,
    \begin{equation}
        \eta = \sqrt{\frac{2R\epsilon }{(Q+1)L_0(f)^3d^2}}\cdot\sqrt{\frac{1}{L_0(\vm_0)\gamma(\vm_0, \nu)}}
    \end{equation}
    And we get
    \begin{equation}
        \frac{1}{Q+1}\sum_{t=0}^Q \mathbb{E}_{\rmU_t, \rk_t}[\|\nabla F_\mu (\vx_t)\|^2] \leq \sqrt{\frac{2L_0(f)^5Rd^2}{(Q+1)\epsilon}}\cdot \sqrt{\gamma(\vm_0, \nu)L_0(\vm_0)^3}
    \end{equation}
    To guarantee the expected squared norm of the gradient of function $F_\mu$ of the order $\delta$, the lower bound for the expected number of queries is
    \begin{equation}
        O(\frac{L_0(f)^5Rd^2}{\epsilon\delta^2}\gamma(\vm_0, \nu)L_0(\vm_0)^3)
    \end{equation}
\end{proof}

\subsection{Proof of Theorem 2}\label{app:proof2}
\begin{proof}
    \begin{equation}
        \begin{aligned}
            P(Sign(H(\vx)) \neq Sign(H_K(\vx))) & \leq P(|H_K(\vx) - H(\vx)| \geq |H(x)|)                                                                                                   \\
                                                & \leq \frac{\mathbb{E}[|H_K(\vx) - H(\vx)|]}{|H(\vx)|}                                                                                     \\
                                                & \leq \frac{\sqrt{\mathbb{E}[(H_K(\vx)-H(\vx)])^2}}{|H(\vx)|}                                                                              \\
                                                & \leq \frac{\sqrt{\mathbb{E}[2(f(\vm_{k_1}(\vx+\mu \ru)) - f(\vm_0(\vx+\mu \ru)))^2 + 2(f(\vm_{k_2}(\vx)) - f(\vm_0(\vx)))^2 ]}}{|H(\vx)|} \\
                                                & \leq \frac{2\nu L_0(f)}{|H(\vx)|}
        \end{aligned}
    \end{equation}
\end{proof}

\section{Accuracy for ImageNet}

\begin{table}[h]
    \centering
    \caption{Evaluation results of PuriDefense and other defense methods on ImageNet under 5 SOTA query-based attacks. The robust accuracy under 200/2500 queries is reported. The best defense mechanism under 2500 queries are highlighted in bold and marked with gray.}
    \resizebox{0.95\textwidth}{!}{
        \begin{tabular}{|c|l|c|c|c|c|c|c|}
            \hline
            Datasets                                                                   & Methods                          & Acc.                                & NES($\ell_\infty$)                                             & SimBA($\ell_2$)                                                & Square($\ell_\infty$)                                          & Boundary($\ell_2$)                                             & HopSkipJump($\ell_\infty$)                                     \\
            \hline
            \multirow{6}{*}{\begin{tabular}{c} ImageNet \\  (ResNet-50)\end{tabular} } & None                             & 76.5                                & 72.9/61.2                                                      & 65.5/50.8                                                      & 37.6/5.2                                                       & 70.7/64.9                                                      & \cellcolor{light-gray}\cellcolor{light-gray}\textbf{68.3/66.0} \\
            \cline{2-8}
                                                                                       & AT~\citep{gowal:2020:uncovering} & 57.5                                & 52.2/51.1                                                      & 54.5/50.7                                                      & 52.9/46.8                                                      & 57.1/57.0                                                      & 57.5/57.3                                                      \\
                                                                                       & FS~\citep{xu:2018:feature}       & \cellcolor{light-gray}\textbf{68.2} & 71.5/59.4                                                      & 61.4/28.2                                                      & 28.8/2.3                                                       & 66.5/60.4                                                      & 64.4/64.4                                                      \\
                                                                                       & IR~\citep{qin:2021:random}       & 64.7                                & \cellcolor{light-gray}\cellcolor{light-gray}\textbf{64.0/63.0} & 61.7/58.3                                                      & \cellcolor{light-gray}\cellcolor{light-gray}\textbf{62.3/60.1} & 65.3/64.9                                                      & 64.8/65.5                                                      \\
                                                                                       & DISCO~\citep{ho:2022:disco}      & 67.7                                & 65.9/60.9                                                      & 61.0/25.7                                                      & 34.5/5.1                                                       & 65.9/63.3                                                      & 67.0/64.6                                                      \\
                                                                                       & \textbf{PuriDefense (Ours)}      & 66.7                                & 65.5/62.9                                                      & \cellcolor{light-gray}\cellcolor{light-gray}\textbf{63.1/61.3} & 65.0/59.1                                                      & \cellcolor{light-gray}\cellcolor{light-gray}\textbf{66.6/65.3} & \cellcolor{light-gray}\cellcolor{light-gray}\textbf{66.2/66.0} \\
                                                                                       & \textbf{PuriDefense-AT (Ours)}   & 57.8                                & 56.0/54.7                                                      & 54.0/53.5                                                      & 55.4/53.2                                                      & 56.8/57.1                                                      & 56.8/56.1                                                      \\
            \hline
        \end{tabular}
    }

\end{table}

\section{Influence on Clean Accuracy\label{app:clean}}
One of the biggest advantage of local implicit purification is that it does not affect the clean accuracy of the model. While the results for evaluation of our mechanism's robust accuracy are shown in \tabref{tab:overall} in \secref{sec:eva}, we also provide the results for clean accuracy in \figref{fig:clean_acc}. Moreover, we have conducted extra experiments on the influence of the numbers of the image patches on the clean accuracy. The results are shown in \figref{fig:clean_acc_matt}. The results are obtained using the whole test set of CIFAR-10 and validation set of ImageNet.

\bheading{Comaprison of Defense Mechanisms.}
We first test clean accuracy on each purification model contained in DISCO and our method. The label name refers to the white-box attacks used to generate adversarial examples for training the purification model. For PuriDefense, a list of the purification model and their according attack and encoder combination can be found in \tabref{tab:purifiers_in_MATT}. For both datasets, all the purification models have a better clean accuracy than adding random noise. Moreover, they all achieve better clean accuracy than adversarially trained models on ImageNet dataset.

\bheading{Influence of the Number of Patches.}
We then test the influence of the number of patches on the clean accuracy. In PuriDefense, we only use image patches for feature encoding and purification. Therefore, the number of patches is a hyperparameter that can be tuned. We test the influence of the number of patches on the clean accuracy. The results are shown in \Figref{fig:clean_acc_matt}. We can see that the clean accuracy is not affected by the number of patches.

\begin{table}[h]
    \centering
    \caption{\small The purification model used in PuriDefense. \label{tab:purifiers_in_MATT}}
    \resizebox{0.9\textwidth}{!}{
        \begin{tabular}[c]{ccccccc}
            \hline
            \textbf{Model Type}      & p0        & p1        & p2         & p3         & p4        & p5        \\
            \midrule
            \textbf{Attack\_Encoder} & BIM\_EDSR & BIM\_RCAN & FGSM\_EDSR & FGSM\_RCAN & PGD\_EDSR & PGD\_RCAN \\
            \hline
        \end{tabular}
    }
\end{table}
\begin{figure}[t]
    \begin{minipage}[t]{0.60\linewidth}
        \centering
        \includegraphics[width=0.9\textwidth]{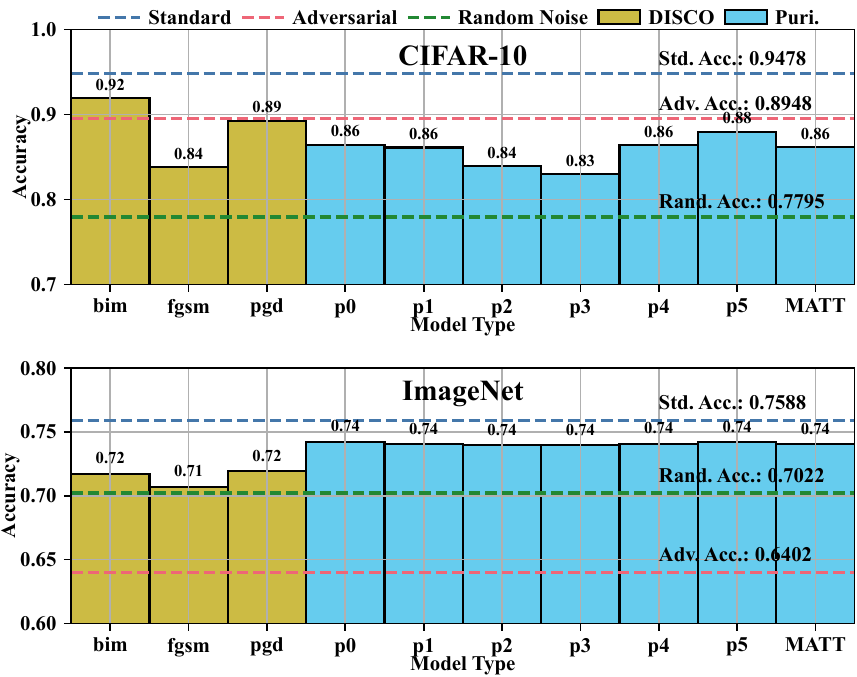}
        \caption{\small Comparison of defense mechanisms and models on clean accuracy. \textbf{Upper Figure:} CIFAR-10 dataset. \textbf{Lower Figure:} ImageNet dataset.}
        \label{fig:clean_acc}
    \end{minipage}
    \hfill
    \begin{minipage}[t]{0.38\linewidth}
        \centering
        \includegraphics[width=0.9\textwidth]{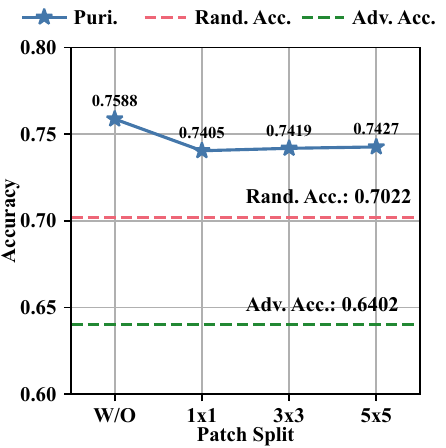}
        \caption{\small Influence of the number of image patches in PuriDefense.}
        \label{fig:clean_acc_matt}
    \end{minipage}
\end{figure}

\review{\section{Attacking Real-world System: A Case Study on Google Vision API}}

\review{
    \subsection{Experiments and Results}
    To illustrate the realistic threat of query-based black-box to real-world commercial systems and verify the effectiveness of our defense mechanisms, we have conducted a real-world attack using an example image on Google Vision API~\cite{google_vision}.

    \bheading{Google Vision API.}
    The google vision API provides the user with a wide range of image analysis services, including label detection, face detection, and text detection. The API is widely used in various applications, thus, the security of the API is crucial.

    \bheading{Attack Settings.}
    To verify the collapse of the system, we first need to formulate the problem. Since the API does not provide the service of image classification service, we adopt the label detection function and formulate the problem. Then, we compared the results of the obtained adversarial image under the straightforwad setting and the protected setting. We conceptually deploy the PuriDefense mechanism on the cloud by adding the purification process before the label detection function, as shown in \textcolor{red}{Figure X1}.

    \lheading{Problem Formulation.}
    The Google Vision API's label detection function returns all the labels along with the scores they detected. For our example image of a dog, the initial results contains more than 5 labels with confidence scores. We denote the image as robust if the top-5 labels contains the label "dog", otherwise we denote the image as adversarial. We start with an initial adversarial image obtained by leveraging the simple binary boundary search methods. Then we use the Boundary Attack~\cite{brendel:2018:decision} to minimize the distance between the adversarial image and original image. The query limit is set to 2,500.

    \bheading{Results.}
    The full attack flow and the results are shown in \textcolor{red}{Figure X2}. We can see that the initial adversarial image found has large discrepancy with the original image. However, after applying the Boundary Attack to attack the unprotected setting, we can obtain an adversarial image with almost imperceptible difference with the original image. With the the protection of PuriDefense, no adversarial image with subtle difference can be found.

    \subsection{Discussion on Real-world Attack}
    Under laboratory conditions, we considered the expense for testing the robustness. Since we have successfully launched an attack on the Google Vision API, we discuss the potential threat. Moreover, the transferability shows that our defense does not rely on the specific data distribution.

    \bheading{Expenses.}
    It is hard to develop specific defense mechanism or attack. We can leverage the general pictures to develop new mechanisms, without the consideration of the dataset distrbution.

    \bheading{Potential Threat.}
    The well known NSFW (Not Safe For Work)

}
\begin{figure}
    \centering
    \includegraphics[width=0.8\textwidth]{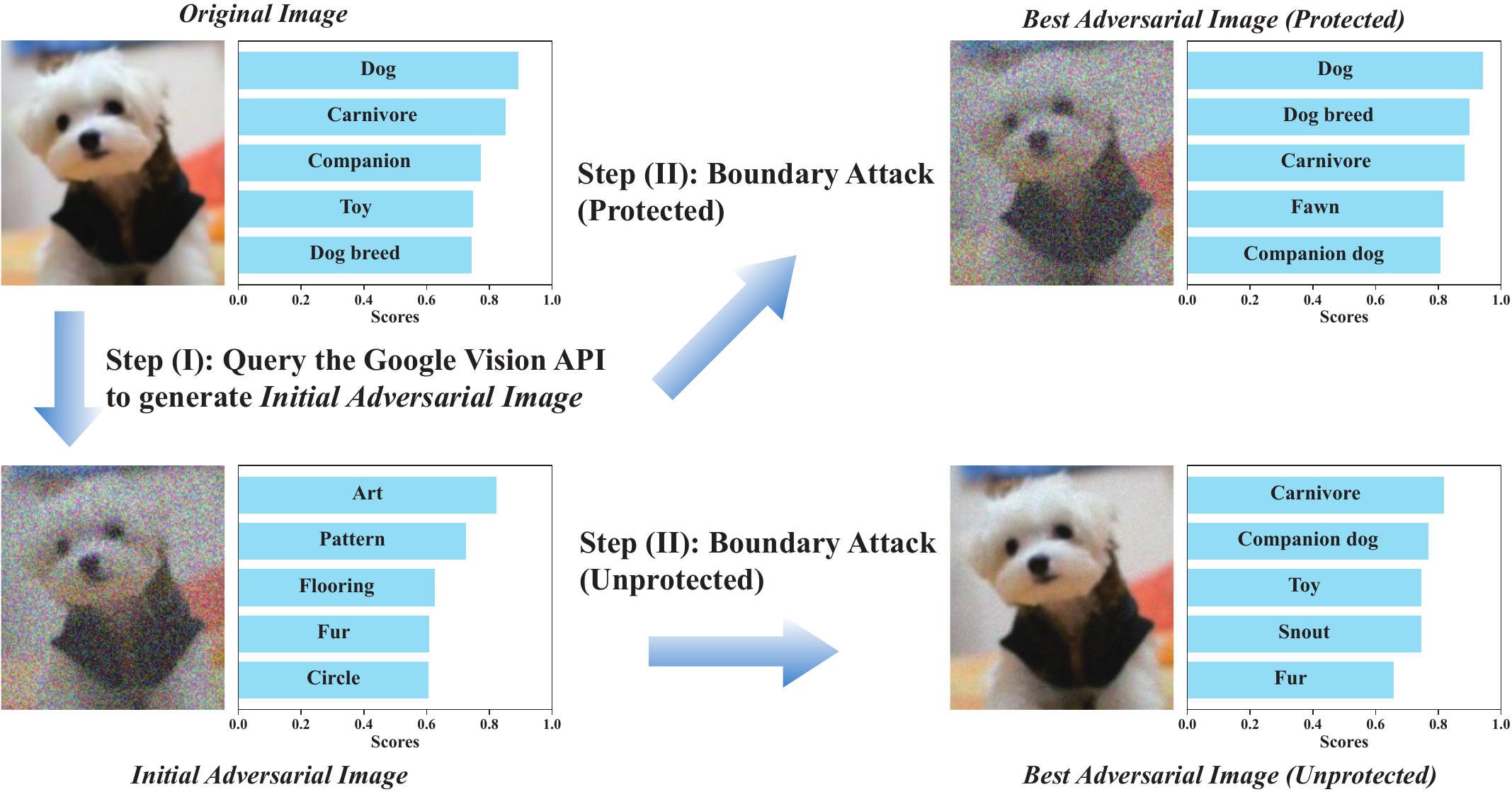}
    \caption{\small This is the real-world attack result.}
\end{figure}